\documentclass[a4paper]{article}
\usepackage[margin=3.9cm]{geometry}

\usepackage{complexity}
\usepackage{mathproblem}
\usepackage{amsmath,amsthm,amssymb}
\usepackage{url}
\usepackage{tikz,pgf}
\usetikzlibrary{shapes,backgrounds,shapes.geometric,shadows,patterns}
\tikzstyle{arg}=[draw,circle,fill=gray!15,inner sep=1pt,minimum size=.5cm]
\tikzstyle{attack}=[->,left,thick,>=stealth]

\usepackage{todonotes}
\usepackage[normalem]{ulem}% \sout command
\newcommand{\F}{F}
\newcommand{\CF}{\mathcal{F}}

\newcommand{\naive}{\mathit{na}}

\newcommand{\cf}{\mathit{cf}}

\newcommand{\csigma}{\operatorname{\mathit{cl-\sigma}}}

\newcommand{\cnaive}{\operatorname{\mathit{cl-naive}}}

\newcommand{\cl}{\mathit{cl}}

\newcommand{\wfCAF}{\mathit{wf}}
\newcommand{\CAF}{\mathit{CAF}}

\newcommand{\U}{U}

\newcommand{\Conc}{\textit{Con}}

\newtheorem{definition}{Definition}

\newtheorem{proposition}{Proposition}

\newtheorem{reduction}{Reduction}

\title{Concurrence for well-formed CAFs:\\Naive Semantics} %\tnoteref{mytitlenote}}
\author{Rafael Kiesel, Anna Rapberger}
\date{\normalsize TU Wien, Institute of Logic and Computation, Austria\\
\{rafael.kiesel,anna.rapberger\}@tuwien.ac.at}

\begin{document}
\maketitle

\begin{abstract}
In the area of claim-based reasoning in abstract argumentation, a claim-based semantics is said to be concurrent in a given framework if all its variants yield the same extensions.
In this note, we show that the concurrence problem with respect to naive semantics is $\coNP$-hard for well-formed CAFs.
This solves a problem that has been left open in \cite{DvorakGRW21}.
\end{abstract}

\section{Introduction}

%\emph{Claim-augmented argumentation frameworks (CAFs)} \cite{DvorakW20} have been introduced to analyze conclusion-focused non-monotonic reasoning problems.
%Formally, CAFs correspond to directed labeled graphs where the nodes correspond to arguments, the labels correspond to the arguments claims, and the arcs indicate (unidirectional) conflicts between them.
\emph{Claim-augmented argumentation frameworks (CAFs)} \cite{DvorakW20} extend \emph{abstract argumentation frameworks (AFs)} \cite{Dung95} by a function that assigns a claim to each argument.
%by enriching argumentation research with graph-theoretical means.
%The compact representation serves as an ideal playground to investigate characteristics of argumentation such as expressiveness \cite{}, dynamics \cite{}, or complexity of reasoning \cite{}. 
Formally, CAFs correspond to directed labeled graphs where the nodes correspond to arguments, the labels correspond to the arguments claims, and the arcs indicate (unidirectional) conflicts between them.
As with arguments in AFs, the acceptance status of claims is decided via \emph{argumentation semantics}. %\cite{Dung95,Verheij96,Caminada06}.
There are often several possibilities to lift AF semantics to claim-level \cite{DvorakRW20c};
%The additional level of claims gives rise to several variants to lift argument-based semantics to claim-level \cite{}. 
considering \emph{naive semantics} which is based on maximization of conflict-free sets gives rise to two claim-based variants: 
The so-called \emph{inherited} variant, which performs maximization on argument-level, and the \emph{claim-level semantics}, which performs maximization over claim-sets.
Deciding whether two different variants of a semantics yield the same claim-based outcome for a given CAF (the so-called \emph{concurrence problem}) can be computationally challenging:
As shown in \cite{DvorakGRW21}, the concurrence problem is $\coNP$-complete with respect to naive semantics.

In this note, we show that deciding concurrence with respect to naive semantics is $\coNP$-hard even for \emph{well-formed CAFs}, an important sub-class of CAFs that impose restrictions on the attack relation.
Well-formed CAFs satisfy a quite natural behavior of conflicts: A CAF is well-formed if arguments with the same claims attack the same arguments. 
This closes the complexity gap in \cite{DvorakGRW21} for the concurrence problem.

\section{Preliminaries}\label{sec:prel}

We introduce abstract argumentation and claim-based reasoning following \cite{Dung95} and \cite{DvorakRW20c}.

\paragraph{Abstract Argumentation.}
%\subsubsection{Abstract Argumentation}
We fix a non-finite background set $\U$. An argumentation framework (AF) \cite{Dung95} is a directed graph $\F = (A,R)$ where $A\subseteq\U$ represents a set of arguments and $R\subseteq A\times A$ models \textit{attacks} between them. 
%We let $E^+_F = \{ a\in A \mid E\text{ attacks }a \}$ for a set $E\subseteq A$.
%In this paper we consider finite AFs only. %and we use $\m{F}$ for the set of all these graphs. 
%For a given $\F = (B,S)$ we let $A(\F) = B$ and $R(F) = S$.
%For $U\subseteq A$ we define the restriction of $\F$ to $U$ as usual, i.e.\ $F\!\!\downarrow_U = (A\cap U , R\cap ({U\times U}))$.
For two arguments $a,b\in A$, if $(a,b)\in R$ we say that $a$ \textit{attacks} $b$ as well as $a$ \textit{attacks} (the set) $E$ given that $b\in E\subseteq A$.
%We frequently use the so-called \textit{range} of a set $E$ defined as $E^\oplus_F = E\cup E^+_F$ where $E^+_F = \{ a\in A \mid E\text{ attacks }a \}$.
%The \textit{$E$-reduct of $\F$} is the AF $F^E = (E^*, R \cap (E^* \times E^*))$ where $E^* = A \setminus E^\oplus$.
%This means, $\F^E$ is the subframework of $\F$ obtained by removing the range of $E$.%, \ie $F^E = F\!\!\downarrow_{A\setminus E^\oplus_F}$. 
A set $E\subseteq A$ is \emph{conflict-free} in $\F$ ($E\in\cf(\F)$) iff for no $a,b\in E$, $(a,b)\in R$. 
%We say 
%$E$ \textit{defends} an argument $a$ if any attacker of $a$ is attacked by some argument of $E$.
%We call a set $E\in \cf(\F)$ is \emph{admissible} in $\F$ ($E\in \adm(\F)$) iff it defends all its elements (we also say, $E$ defends itself).
A \emph{semantics} is a function $\sigma$ with $F\mapsto\sigma(F)\subseteq 2^A$.
%This means, given an AF $\F = (A,R)$ a semantics returns a set of subsets of $A$. These subsets are called $\sigma$-\emph{extensions}.
In this note we focus on \emph{naive semantics}.
\begin{definition}
For an AF $\F=(A,R)$, a set $E\subseteq A$ is \emph{naive} ($E\in \naive(\CF)$) iff $E$ is $\subseteq$-maximal in $\cf(\F)$.
\end{definition}

\paragraph{Reasoning about claims.}
%\subsubsection{Reasoning about Claims}
A \emph{claim-augmented argumentation framework (CAF)} \cite{DvorakW20} is a triple $\CF=(A,R,\cl)$
where $F=(A,R)$ is an AF and $\cl$
%$\cl:A\rightarrow \mathcal{C}$ 
is a function which assigns a claim to each argument in $A$.
%; $\mathcal{C}$ is a set of possible claims.
The claim-function is extended to sets in the natural way, i.e., for a set $E\subseteq A$, we let $\cl(E)=\{\cl(a)\mid a\in E\}$.
A CAF is \emph{well-formed} iff $a^+_{F}=b^+_{F}$ for all $a,b\in A$ with $\cl(a)=\cl(b)$, i.e., arguments with the same claim attack the same arguments.
The literature offers several ways to extend semantics for AFs that involve claims respectively arguments to a different extent. 
%Inherited semantics as introduced in \cite{} first evaluate the underlying AF before inspecting the claims of the arguments in the extensions, i.e., for a semantics $\sigma$, its inherited variant is defined as $\sigma\!_c(\CF)=\{\cl(E)\mid E\in\sigma(\F)\}$. 
We introduce \emph{inherited} and \emph{claim-level semantics} for naive semantics that perform maximization in different stages of the evaluation. 
\begin{definition}
For a CAF $\CF=(A,R,\cl)$, $F=(A,R)$, and a semantics $\sigma$, we let $\sigma\!_c(\CF)=\{\cl(E)\mid E\in\sigma(\F)\}$. A set $S\subseteq \cl(A)$ is
\begin{itemize}
\item  \emph{i-naive} ($S\in \naive\!_c(\CF)$) iff $S\in\naive\!_c(\CF)$, i.e., there is $E\subseteq A$ with $S = \cl(E)$ and $E$ is $\subseteq$-maximal in $\cf(F)$;
\item \emph{cl-naive} ($S\in \cnaive(\CF)$) $S$ is $\subseteq$-maximal in $\cf\!_c(\CF)$.
\end{itemize} 
%
%
%($S\in \naive\!_c(\CF)$) is \emph{i-naive} iff there is $E\subseteq A$ 
%with $E\in \naive(F)$ \todo{and $S = \cl(S)$, oder?}; $S$ is \emph{cl-naive} ($S\in \cnaive(\CF)$) $S$ is $\subseteq$-maximal in $\cf\!_c(\F)$.
\end{definition}
As shown in \cite{DvorakRW20c}, $\cnaive(\CF)\subseteq \naive\!_c(\CF)$.

\section{The Concurrence Problem for Naive Semantics}
An interesting problem that arises when considering different variants of semantics for CAFs is the so-called \emph{concurrence problem}.
\begin{center}
	\vspace{-0.2cm}
	\begin{mathproblem}{$\Conc_{\sigma}^\Delta$, $\Delta\in \{\CAF,\wfCAF\}$}
		Input: & A CAF $\CF$ (if $\Delta=\CAF$)/a well-formed CAF $\CF$ (if $\Delta=\wfCAF$)\\
		Output: & \textsc{true} iff %$\rho(\CF)=\rho(\CG)$, i.e., 
		$\sigma\!_c(\CF)=\csigma(\CF)$\\
	\end{mathproblem}
\end{center}
%\begin{definition}[Concurrence] For a semantics $\sigma$ and $\Delta\in \{\CAF,\wfCAF\}$, we let $\Conc_{\sigma}^\Delta$ denote the following problem: Given a 
%	CAF $\CF$, does it hold that $\sigma\!_c(\CF)=\csigma(\CF)$?
%\end{definition}
For naive semantics, the problem can be formulated as follows: Given a CAF $\CF$, is it the case that maximization on argument-level yields the same accepted sets as maximization on claim-level?

In \cite{DvorakGRW21} it has been shown that $\Conc_{\sigma}^\CAF$ is $\coNP$-complete. Since concurrence for well-formed CAFs with respect to naive semantics is a special case of CAFs we obtain upper bounds for $\Conc_{\sigma}^\wfCAF$, that is, $\Conc_{\sigma}^\wfCAF$ is in $\coNP$.

We restate the following proposition \cite{DvorakGRW21}.
\begin{proposition}
\label{lem:incomp_pref_naive}
For a CAF $\CF=(A,R,\cl)$,
$\naive\!_c(\CF)=\cnaive(\CF)$ if and only if $naive\!_c(\CF)$ is incomparable.
\end{proposition}
%\begin{proof}
%Let $\sigma=\pref$ (the proof for $\sigma=\naive$ is analogous).
%Assume $\pref\!_c(\CF)$ is incomparable and let $S\in \pref\!_c(\CF)$.
%Then $S\in \adm\!_c(\CF)$. Now assume there is $T\in \adm\!_c(\CF)$ 
%with $T\supset S$. Consider a $\adm$-realization $E$ of $T$ in $\CF$
%and let $E'\in \pref((A,R))$ with $E\subseteq E'$. But then 
%$\cl(E')\in \pref\!_c(\CF)$ and $\cl(E')\supseteq
%T\supset S$, contradiction to $\pref\!_c(\CF)$ being incomparable.
%\end{proof} 
Thus it suffices to verify incomparability of $\naive\!_c(\CF)$.
An $\NP$ procedure for the complementary problem is by a standard guess and check procedure:
Guess $E,G\subseteq A$ and check (i) $E,G\in \sigma((A,R))$ and (ii) $\cl(E)\subset \cl(G)$. 
The former can be checked in time polynomial in the number of arguments in $\CF$.

Next we show that verifying incomparability of $\naive\!_c(\CF)$ is $\NP$-hard even if $\CF$ is well-formed. We will first define the base reduction, which we slightly extend to obtain the the reduction for the well-formed case.

\begin{reduction}
\label{red:wf}
Let $\varphi$ be given by a set of clauses 
$C=\{cl_1, \dots, cl_n\}$ over
atoms in $X$. 
We construct $(A,R,\cl)$ with
%\begin{itemize}
%\item 
%$A=X\cup \bar X \cup C\cup \{\varphi\} \cup \{a_1, a_2\}$, with $\bar X=\{\bar x\mid x\in X\}$;
%\item 
\begin{align*}
A = & X\cup \bar X \cup C\cup \{\varphi\} \cup \{a_1, a_2\}, \text{ with } \bar X=\{\bar x\mid x\in X\},\\
R  =	& \{(x,cl)\mid cl\in C,x\in cl\} \cup \{(\bar x,cl)\mid cl\in C,\neg x\in cl\}\\
		&\cup \{(x,\bar x),(\bar x, x)\mid x\in X\} \cup \{(cl_i,\varphi)\mid i\leq n\} \cup \{(\varphi, a_2)\},
\end{align*}
%$\begin{array}{rcl}
%A& = & X\cup \bar X \cup C\cup \{\varphi\} \cup \{a_1, a_2\}, \text{ with } \bar X=\{\bar x\mid x\in X\},\\
%R 	& =	& \{(x,cl)\mid cl\in C,x\in cl\} \cup \{(\bar x,cl)\mid cl\in C,\neg x\in cl\}\\
%	& 	&\cup \{(x,\bar x),(\bar x, x)\mid x\in X\} \cup \{(cl_i,\varphi)\mid i\leq n\} \cup \{(\varphi, a_2)\}.
%\end{array}
%$
%\end{itemize}
and $\cl(x)=x$, $\cl(\bar x)=\bar x, \cl(cl_i)=cl_i, \cl(\varphi) = \varphi$ and $\cl(a_i) = a$.
\end{reduction}
An example of this reduction is given in Figure~\ref{fig:red_wf}. 
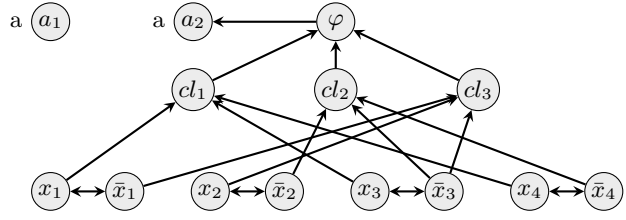
\begin{figure}[t]
\centering
\begin{tikzpicture}[>=stealth,xscale=0.75, yscale=0.5]
\small
		\path (0,0.5)
			node[arg](phi) {$\varphi$};
		\node[arg,label={left:a}] (a2) at (-2.5,0.5) {$a_2$};
		\node[arg,label={left:a}] (a1) at (-5, 0.5) {$a_1$};
		\draw[->, thick] (phi) -- (a2);
		\path 	(-2.5,-1.3) 
			node[arg](c1){$cl_1$}
			++(2.5,0) node[arg](c2){$cl_2$}
			++(2.5,0) node[arg](c3){$cl_3$};
		\path 	(-5,-4) 
					node[arg](z1){$x_1$}
			++(1.3,0) node[arg](nz1){$\bar x_1$}
			++(1.5,0) node[arg](z4){$x_2$}
			++(1.3,0) node[arg](nz4){$\bar x_2$}
			++(1.5,0) node[arg](z2){$x_3$}
			++(1.3,0) node[arg](nz2){$\bar x_3$}
			++(1.5,0) node[arg](z3){$x_4$}
			++(1.3,0) node[arg](nz3){$\bar x_4$}
			;
		\path [->,thick]
			(c3) edge (phi)
			(c2) edge (phi)
			(c1) edge (phi)
			;
		\path [left,->, thick]
			(z1) edge (c1)
			(z2) edge (c1)
			(z3) edge (c1)
			(nz2) edge (c2)
			(nz3) edge (c2)
			(nz4) edge (c2)
			(nz1) edge (c3)
			(nz2) edge (c3)
			(z4) edge (c3);
		\path [left,<->, thick]
			(z1) edge (nz1)
			(z2) edge (nz2)
			(z3) edge (nz3)
			(z4) edge (nz4);
\end{tikzpicture}
    \caption{  
    Reduction~\ref{red:wf} for a formula $\varphi$ which is
given by the clauses 
      $\{\{x_1,x_3,x_4\},\{\bar{x}_3 , \bar{x}_4 , \bar{x}_2)\},\{\bar{x}_1, \bar{x}_3 , x_2\}\}$.}\label{fig:red_wf}
\end{figure}

The intuition behind this reduction is the following: 
We take the smallest sub-CAF that is not concurrent.
This corresponds to $(\{a_1, a_2, \varphi\}, \{(\varphi, a_2)\},\cl)$ with $\cl(a_i) = a$ and $\cl(\varphi) = \varphi$, with naive extensions $\{a_1, \varphi\}$ and $\{a_1, a_2\}$. Observe that the second extension requires that $\varphi$ is not taken. We then add the additional part corresponding to a propositional formula in CNF that only allows $\varphi$ to not be included without additional changes, when the formula is satisfiable. Thus allowing us to reduce the concurrence problem to UNSAT.

In the example in Figure~\ref{fig:red_wf} we see that $\{a_1, \varphi, x_1, x_2, x_3, \bar{x_4}\}$ and $\{a_1, a_2, x_1, x_2, x_3, \bar{x_4}\}$ are naive extensions because $\{x_1, x_2, x_3, \bar{x_4}\}$ is a satisfying assignment of the formula.

Generally, the proof proceeds as follows.
\begin{proposition}
$\Conc_\naive^\wfCAF$ is $\coNP$-hard.
\end{proposition}
\begin{proof}
For hardness, we present a reduction from UNSAT: Let $\varphi$ be given by a set of clauses $C=
\{cl_1,\dots,cl_n\}$ over literals in $X$. W.l.o.g.\ we can assume that $\varphi$ does not contain tautological clauses, i.e., 
there is no $cl_i$, $i\leq n$ with $x,\bar x\in cl_i$ for any $x\in X$. Let $(A,R,\cl)$ be defined as in Reduction~\ref{red:wf}. 
%We construct a CAF $\CF=(A,R,\cl)$ with $\cl(x)=x$, $\cl(\bar x)=\bar x, \cl(cl_i)=cl_i, \cl(\varphi) = \varphi$ and $\cl(a_i) = a$. 
We will show $\varphi$ is unsatisfiable iff $\naive\!_c(\CF)$ is incomparable.

First asssume $\varphi$ is satisfiable and consider  a model $M$ of $\varphi$.
Let $E=M\cup \{\bar x\mid x\notin M\} \cup \{\varphi, a_1\}$. Clearly, $E$ is conflict-free; moreover, as $M$ satisfies each clause $cl_i$ there is either $x\in cl_i$ with $x\in M$ 
or $\bar x\in cl_i$ with $x\notin M$, thus $E$ attacks each $cl_i$. Since, also $\varphi$ and $a_1$ are in $E$ we the only argument left to consider is $a_2$, which is however attacked by $\varphi$ and can therefore not be included, while preserving conflict-freeness. We can conclude that $E$ is a subset-maximal conflict-free set. 
Moreover, $E' = M\cup \{\bar x\mid x\notin M\} \cup \{a_1, a_2\}$ is also a subset-maximal conflict-free set, since still every $cl_i$ is attacked and $\varphi$ attacks $a_2$.
It follows that $\naive\!_c(\CF)$ is not incomparable since
$\cl(E)= M\cup \{\bar x\mid x\notin M\} \cup \{\varphi, a\}$ is a strict superset of $M\cup \{\bar x\mid x\notin M\} \cup \{a\} = \cl(E')$ and both are contained in $\naive\!_c(\CF)$.

Now assume that $\varphi$ is unsatisfiable. Let $E$ be a subset-maximal conflict-free set. If $\varphi \in E$ it follows that none of the $cl_i$ are in $E$. Therefore, it holds that for each $x$ exactly one of $x$ and $\bar{x}$ is in $E$. Furthermore, it always holds that $a_1$ is in $E$. This means that $\cl(E) = \{\varphi, a\} \cup \{x \mid x \in X'\} \cup \{\bar{x} \mid x \not\in X'\}$ for some subset $X'$ of $X$. 

Case I: Assume there is a naive extension $E'$ such that $\cl(E) \subsetneq \cl(E')$. Then $E'$ must contain $\cl_i$ for some $i$, $x$ for some $x \not\in X'$ or $\bar{x}$ for some $x \in X'$. This is not possible, since this would imply $E \subsetneq E'$, which is a contradiction to the assumption that $E$ and $E'$ are naive extensions.

Case II: Assume there is a naive extension $E'$ that contains $\varphi$ such that $\cl(E') \subsetneq \cl(E)$. Then there must be some $x$ with $x \in X'$ or $\bar{x}$ with $x \not\in X'$ that is not contained $E'$. This is not possible, since this would imply $E' \subsetneq E$, which is a contradiction to the assumption that $E$ and $E'$ are naive extensions.

Case III: Assume there is a naive extension $E'$ that does not contain $\varphi$ such that $\cl(E') \subsetneq \cl(E)$. This is impossible, as $\varphi$ is unsatisfiable, which entails that there is a clause $cl_i$ that is not satisfied by $X'$. Therefore, $cl_i$ is not attacked by $E \setminus \varphi$ and also not $E'$. Thus $E'$ contains at least one $cl_i$, which means that also $\cl(E')$ contains $cl_i$.

We see that any naive extension $E$ that contains $\varphi$ is incomparable to any other naive extension on claim level.

Next we show that the same holds for any naive extension $E$ that does not contain $\varphi$. It follows that $E$ contains $a_1, a_2$. Since $\CF' = (A \setminus \{a_1, a_2, \varphi\}, R \setminus \{(\varphi, a_2)\} \cup \{(cl_i, \varphi) \mid i = 1, \dots , n\}, \cl)$, the rest of the CAF, is such that all arguments have distinct claims, there is a one to one correspondence between $\naive\!_c(\CF')$ and $\naive(\CF')$, which implies that they are all incomparable with one another.
\end{proof}

\section{Conclusion}
We see that even for well formed CAFs the concurrence problem is $\coNP$-complete, as it is the case for general CAFs. This shows that while well formedness is an interesting property that only allows a fragment of the CAFs that might be deemed more reasonable, it does not lead to CAFs that are simpler, at least with respect to the concurrence problem.

%Speaking in graph-theoretical terms, we have shown that deciding whether the maximal independent sets (in terms of the labels of the nodes) of a directed labeled graph where nodes with the same label have the same outgoing edges are incomparable is $\coNP$-complete.
%While the direction of the arcs do not play a role for in the context of independent sets, we observe that the well-formedness condition poses certain restrictions on the structure of the graph.

\section*{Acknowledgements}
%\todo{We need this, right?}
This work has been supported by the Austrian Science Fund (FWF) Grant W1255-N23.

\bibliography{argumentation}
\end{document}